\documentclass{fundam}

\usepackage{url} % takes care of hyperlinks, preferred over hyperref
\usepackage[ruled,lined]{algorithm2e}% provides Algorithm environment
\usepackage{graphicx}% allows for inclusion of graphic files (figures)
\usepackage{tikz}
\usetikzlibrary{automata, positioning, arrows}
\usepackage{mymacros,amssymb,amsfonts}

\begin{document}

%%%%%%%  parameters to be filled in by copy-editor  %%%%%%%%%%

\setcounter{page}{175}
\publyear{22}
\papernumber{2125}
\volume{186}
\issue{1-4}

        \finalVersionForARXIV
            %%\finalVersionForIOS

\title{Resource Bisimilarity in Petri Nets is Decidable}

\author{Irina A.\ Lomazova\thanks{Address for correspondence: ialomazova@gmail.com}%\newline
                         %  \hspace*{-1mm}\corresponding%$^C$Corresponding author}$^C$
\\
      ialomazova@gmail.com
  \and  Vladimir A. \ Bashkin
         %%Yaroslavl, Russia  \\
          %%bas@uniyar.ac.ru
 \and Petr Jan\v{c}ar\\
         Dept.$\,$of Computer Science,  Faculty of Science\\
         Palack\'y University, Olomouc, Czechia \\
          	petr.jancar@upol.cz
	}

\maketitle

\runninghead{I. Lomazova et al.}{Resource Bisimilarity in Petri Nets is Decidable}

\vspace*{4mm}
\begin{abstract}
Petri nets are a popular formalism for modeling and analyzing distributed systems.
	Tokens in Petri net models can represent the control flow state or resources produced/consumed by transition firings.
We define a resource as a part (a submultiset) of Petri net markings and call two resources equivalent when replacing one of them with another in any marking does not change the observable Petri net behavior.
We consider resource similarity and resource bisimilarity, two
	congruent restrictions of bisimulation equivalence on Petri
	net markings. Previously it was proved that resource
	similarity (the largest congruence included in bisimulation
	equivalence) is undecidable. Here we present an algorithm for
	checking resource bisimilarity, thereby proving that this
	relation  (the largest congruence
	included in bisimulation equivalence that is a bisimulation)
is decidable.
We also give an example of two resources in a Petri net that are
	similar but not bisimilar.
\end{abstract}


\eject

\section{Introduction}

The concept of process equivalence can be formalized in many different ways \cite{GLABBEEK2001}.
One of the most important is \emph{bisimulation equivalence}~\cite{Park1981, Milner1989}, which captures the main features of the observable behavior of a system.
Generally, bisimulation equivalence, also called \emph{bisimilarity}, is defined as a relation on sets
of states of  labeled transition systems (LTS). Two states are bisimilar if their behavior cannot be distinguished by an external  observer.

Petri nets are a popular formalism for modeling and analyzing  distributed systems, on which many notations used in practice are based.
In Petri nets, states are represented as markings --- multisets of
tokens residing in Petri net places. The interleaving semantics
associates a labeled Petri net (in which transitions are labeled with
actions) and its initial marking with the corresponding LTS that
describes the behavior of the net. For Petri nets, many important
behavioral properties, e.g.\ reachability, are decidable,
but behavioral equivalences like bisimilarity are undecidable~\cite{Jancar1995}.

As a mathematical formalism, Petri nets are equivalent to the subclass (\textbf{P},\textbf{P})-PRS of Process Rewrite Systems (PRS) introduced by R.~Mayr \cite{Mayr2000}. PRS  allow   a (possibly infinite) LTS to be represented by a finite set of rewrite rules using algebraic operations of sequential and parallel composition.
Several well-studied classes of infinite-state systems are included in the PRS hierarchy, which is based on operations used in rewrite rules to define the specific PRS:
basic process algebras (BPA) are (\textbf{1},\textbf{S})-PRS,
basic parallel processes (BPP) are (\textbf{1},\textbf{P})-PRS,
pushdown automata (PDA) are (\textbf{S},\textbf{S})-PRS,
Petri nets are (\textbf{P},\textbf{P})-PRS and process algebras (PA) are (\textbf{1},\textbf{G})-PRS.

The PRS hierarchy is strict (regarding expressivity) and sets a number
of interesting decidability borders, in particular for bisimilarity.
While bisimilarity is undecidable for (\textbf{P},\textbf{P})-PRS,
corresponding to Petri nets, it is decidable for
the class (\textbf{1},\textbf{P})-PRS~\cite{Christensen1993},
corresponding
to labeled communication-free Petri nets, in which each transition has
no more than one input place. Due to the restriction on the structure
of communication-free Petri nets, bisimilarity for them is a congruence, and this property
was crucial for the decidability proof in~\cite{Christensen1993}.
(This property is also implicitly used in the later
proof~\cite{Jancar22}.)

\medskip

In Petri net models, places can be interpreted as resource repositories, and tokens  represent resource availability or resources themselves. A transition firing, in turn, can be thought of as resource handling, relocation, deletion, and creation. Hence, Petri nets have a clear resource perspective and are widely used to model and analyze various types of resource-oriented systems such as production systems, resource allocation systems, etc.

Since bisimulation equivalence cannot be effectively checked for Petri
nets, it is natural either to restrict the model (for example, to consider  communication-free or one-counter Petri nets), or to strengthen the equivalence relation.
In this paper we explore the second option.  We consider resources as parts of Petri net markings (multisets of tokens residing in Petri net places) and study the possibility of replacing  one resource  with another  in any  marking without changing the observable behavior of the net.

\medskip
\noindent \emph{\textbf{Related work.} }\noindent \smallskip

Numerous studies have been devoted to various aspects of resources in Petri nets. We name just a few of them.

In open nets special resource places are used for modeling a resource interface of the system   \cite{Baldan2015,Heckel2003,Dong2016,Lomazova2013}. In workflow nets resource places represent resources that can be consumed and produced  during the business process execution.  Obviously, resource places not only demonstrate a resource flow, but may substantially influence the system control flow \cite{Bashkin2014,Sidorova2013}.

As a mathematical formalism, Petri nets are closely connected with
Girard's linear logic~\cite{Girard1987}, which also has a nice resource interpretation, and for different classes of Petri nets it is possible to express a Petri net as a linear logic formula
\cite{Farwer1999,Farwer2001}.

\medskip
We have already mentioned bisimulation equivalence, which is defined
as the largest bisimulation in a given LTS.
In \cite{Olderog1989}, an equivalence on places of a Petri net was defined which when lifted to the reachable markings preserves their token game and distribution over the places.
This structural equivalence is called ``strong
bisimulation'' in \cite{Olderog1989}\footnote{Not related to the more common use of terms ``strong bisimulation'' and ``weak bisimulation'' for fundamental state bisimulation and state bisimulation in LTS with invisible (silent) transitions.},  because it is a non-trivial subrelation of bisimilarity.

The term \emph{place bisimulation} comes from \cite{Autant1991} where
Olderog's concept was improved and used for a reduction of nets.
In a Petri net, two bisimilar places can be fused without changing the observable net behavior.

In \cite{Autant1992}, a polynomial algorithm for computing the largest place bisimulation was presented.
Equivalences on  sets of places were further explored in \cite{Autant1994, Pfister1995, Schnoebelen2000}.
In particular, it was proven (using the technique from \cite{Jancar1995}) that a weaker relation, called the  \emph{largest correct place fusion}, is undecidable.
Places $p$ and $q$ can be correctly fused if for any marking $M$ the two markings $M+p$ and $M+q$ are bisimilar.

Independently,  in \cite{Voorhoeve1996} a similar equivalence relation, called \emph{structural bisimilarity}, was studied for labeled Place-Transition Nets.
This relation is defined on the set of all vertices of a Petri net (both places and transitions) and also uses a kind of weak transfer property.

In \cite{Gorrieri2020Team, Gorrieri2020AStudy}, a notion of \emph{team bisimulation} was defined for communication-free Petri nets.
Two distributed systems, each composed of a team of sequential non-cooperating agents (tokens in a communication-free net), are equivalent if it is possible to match each sequential component of one
system with a team-bisimilar sequential component of another system (as in sports, where competing teams have the same number of equivalent player positions).
For Petri nets, the team bisimulation coincides with the place bisimulation from \cite{Autant1991}.

In \cite{Bashkin_FI2003}, a new equivalence on
 Petri net resources, called \emph{resource similarity} was proposed.
As already mentioned, a~resource in a Petri net is a part of its
marking, and
two resources are similar if replacing one of them by another in any
marking does not change the observable net behavior; in our case it
means that the behavior remains in the same class of bisimulation
equivalence.
It was shown
that the correct place fusion equivalence from~\cite{Schnoebelen2000}
is a special case of  resource similarity, namely the place fusion is
resource similarity for one-token resources. Hence the
undecidability result for the place fusion extends to resource
similarity.
On the other hand,
resource similarity is a congruence and thus can be generated by a finite basis. A special type of minimal basis,
called the \emph{ground basis}, was presented in \cite{Bashkin_FI2003}.

In~\cite{Bashkin_FI2003} also a stronger equivalence  of  resources in
a Petri net, called \emph{resource bisimilarity}\footnote{Not related
to the notion of ``resource bisimulation'' from \cite{Corradini1999}.}
was defined;
an equivalence of resources was called a resource bisimulation if its additive and transitive closure is a bisimulation.
 Properties of the mentioned resource equivalences were studied in \cite{Bashkin_PACT2003,Bashkin2005}.
In particular, it was shown that resource bisimulation is defined by a
weak transfer property, similar to the weak transfer property of place
bisimulation. It was proven  that there exists the largest resource
bisimulation, and it coincides with resource bisimilarity.
The questions of whether the resource bisimilarity is decidable and
whether it is strictly stronger than resource similarity remained open
in~\cite{Bashkin_PACT2003,Bashkin2005,10.1007/978-3-319-57861-3_3}.

\medskip\noindent
\emph{\textbf{Our contribution.}}\smallskip

In  this paper we give  answers to these open questions.
Based on the well-known ``tableau method'' (see, e.g.,
\cite{Aceto2011}), we construct an algorithm for checking resource
bisimilarity in Petri nets. Thus, it is
 proved that resource bisimilarity is decidable in Petri nets and
 is strictly stronger than resource similarity.
 We also give an example of a~labeled Petri net in which two similar resources are not resource bisimilar.

Interestingly, for resource similarity and resource bisimilarity
we have that both  are congruences, and thus both have finite bases, but one of them is decidable, and the other is not.

\medskip
\emph{Organization of the paper.} In Section~\ref{sec:prelim} we
recall basic notions and notations, including the result on
finitely-based congruences. Section~\ref{sec:ressimbisim} clarifies
how resource bisimilarity, resource similarity, and bisimilarity are
related.  Section~\ref{sec:algor} shows the announced algorithm
deciding resource bisimilarity, and Section~\ref{sec:conclusions}
provides conclusions and additional remarks.

\section{Preliminaries}\label{sec:prelim}

By $\Nat$ and $\Nat_+$ we denote the sets  of non-negative integers
and of positive integers, respectively.

\paragraph{Multisets.}
A \emph{multiset} $m$ over a set $S$ is a mapping $m:
S\rightarrow\Nat$; hence
a multiset may contain several copies of the same element.
By $\mathcal{M}(S)$ we denote the set of multisets over $S$ (which
could be also denoted by $\Nat^S$).
The inclusion $\subseteq$ on $\mathcal{M}(S)$
coincides with the component-wise order $\leq$\,:
we put $m\subseteq m'$, or $m\leq m'$, if
 $\forall s\in S: m(s) \leq m'(s)$.
By $m=\emptyset$ we mean that $m(s)=0$ for all $s\in S$.

Given $m,m'\in\mathcal{M}(S)$,
the union $m\cup m'$
is generally different than the sum $m+m'$: for each $s\in S$ we have
$(m\cup m')(s)=\max\,\{m(s),m'(s)\}$ and
$(m+m')(s)=m(s)+m'(s)$.
We also define the multiset subtraction $m-m'$:
 for each $s\in S$ we have
$(m - m')(s)=\max\,\{\,m(s)-m'(s),0\,\}$.

\medskip
In this paper we only deal with multisets over finite sets $S$.
In this case the cardinality $|m|$ of each $m\in\mathcal{M}(S)$ is
finite: we have $|m|=\sum_{s\in S}m(s)$.
We note that the partial order $\leq$ on $\mathcal{M}(S)$ can be
naturally extended to a total order, the
\emph{cardinality-lexicographic} order $\sqsubseteq$\,:
\begin{center}
we put $m\sqsubseteq m'$ if $|m|<|m'|$, or $|m|=|m'|$ and
$m\leqlex m'$,
\end{center}
	where $\leqlex$ is a lexicographic order.
More precisely, the \emph{lexicographic order} $\leqlex$ on $\mathcal{M}(S)$
assumes that the elements of $S$ are ordered, in which case
$m\in\mathcal{M}(S)$ can be also viewed as a~vector
$m=((m)_1,(m)_2,\dots,(m)_{|S|})$
in $\Nat^{|S|}$; we put
$m\leqlex m'$ if $m=m'$ or
 $(m)_i<(m')_i$ for the least $i$ for which $(m)_i$ and $(m')_i$
differ.

\paragraph{Labeled Petri nets.}

A \emph{Petri net} is a tuple $N=(P,T,W)$  where $P$ and $T$
are finite disjoint sets of \emph{places} and \emph{transitions},
respectively, and $W:(P\times T)\cup(T\times P)\to\Nat$ is an
\emph{arc-weight function}.
A \emph{labeled Petri net} is a tuple $N=(P,T,W,l)$ where
$(P,T,W)$ is a Petri net and $l:T\to\Act$ is a~labeling function,
mapping the transitions to \emph{actions} (observed events) from a set
$\Act$.
Figure~\ref{fig:buying} shows an example of a labeled Petri net, with
three transitions depicted by boxes, all being labeled with the
same action $b$; the places are depicted by circles, and
the arc-weight function is presented by (multiple) directed arcs
(there is no arc from $p\in P$ to $t\in T$ if $W(p,t)=0$, and
similarly no arc from $t$ to $p$ if $W(t,p)=0$).

A \emph{marking} in a Petri net  $N=(P,T,W)$
is a function $M: P\to \Nat$, hence a multiset $M\in\Mpp$;
it gives the number of \emph{tokens} in each place.
Tokens residing in a place are often interpreted as resources of
some type, which are consumed or produced by transition firings (as
defined below).

For a transition $t\in T$,
the \emph{preset} $\pre{t}$ and the
\emph{postset} $\post{t}$ are defined as the multisets over $P$
such that $\pre{t}(p)= W(p,t)$ and $\post{t}(p)= W(t,p)$ for each
$p\in P$. A \emph{transition} $t\in T$ \emph{is enabled in a marking} $M$
if $\pre{t}\leq M$; such a \emph{transition may fire in} $M$, yielding
the marking
$M'=(M-\pre{t})+\post{t}$ (hence $M'(p)= M(p)-W(p,t)+W(t,p)$ for each
$p\in P$). We denote this firing by $M\tu{t}M'$.

\paragraph{Bisimulation equivalence (in labeled Petri nets).}
Informally speaking, two markings (states) in a labeled Petri net are considered equivalent if they generate the same observable behavior.
Finding equivalent states may be very helpful for reducing the state space
 when analyzing behavioral properties of a given net.
A classical behavioral equivalence is bisimulation equivalence, also
called bisimilarity~\cite{Milner1989}. We recall its definition in the
framework of labeled Petri nets; below we thus implicitly assume a~fixed
labeled Petri net $N=(P,T,W,l)$.

\medskip
We say that a
\emph{relation} $R\subseteq \Mpp\times\Mpp$ \emph{satisfies} the \emph{\it
transfer property w.r.t. a relation} $R'\subseteq \Mpp\times\Mpp$
if for each pair $(M_1,M_2)\in R$ and for each
firing $M_1\tu{t}M_1'$ there exists an (``imitating'') firing
$M_2\tu{u}M_2'$ such that $l(u)=l(t)$  and $(M_1',M_2')\in R'$.
We can illustrate this property by the following diagram.

$\qquad\qquad\qquad\qquad\qquad\qquad M_1\qquad
R \qquad M_2$

$\qquad\qquad\qquad\qquad\qquad\qquad~\downarrow
t\quad\qquad\qquad\downarrow (\exists)u: l(u)=l(t)$

$\qquad\qquad\qquad\qquad\qquad\qquad M_1'\qquad R' \qquad M_2'$
\\
By saying that $R$ \emph{has the transfer property} we mean that
$R$ satisfies the transfer property w.r.t. itself (hence $R'=R$ in the
diagram).

\medskip
A relation $R\subseteq \Mpp\times\Mpp$ is a \emph{bisimulation} if
$R$ has the transfer property and also $R^{-1}$ has the transfer property.
Markings $M_1$ and $M_2$ are  \emph{bisimilar} (or \emph{bisimulation
equivalent}), written $M_1 \sim M_2$,
if there exists a bisimulation $R$ such that $(M_1,M_2)\in R$; the
relation $\sim$, called \emph{bisimilarity}
(or \emph{bisimulation equivalence}),
 is thus the union of all bisimulations.
It is easy to check that bisimilarity is an equivalence
(a reflexive, symmetric, and transitive relation), and that it is
itself a~bisimulation.
Hence the relation $\sim$ is the largest bisimulation (related to our
fixed labeled Petri net $N=(P,T,W,l)$).

\medskip
It is also standard to define
\emph{stratified bisimilarity relations}~\cite{Hennessy1985}
$\sim_i\subseteq \Mpp\times \Mpp$, for $i\in\Nat$:
\begin{itemize}
\itemsep=0.9pt
\item
	$\sim_0\mathop{=}\Mpp\times\Mpp$
		(hence $M_1\sim_0 M_2$ for all $M_1,M_2\in\Mpp$);
\item
	$\sim_{i+1}$ is the largest symmetric relation that satisfies
		the transfer property w.r.t. $\sim_i$
		(hence
	$M_1\sim_{i+1}M_2$ iff for each $M_1\tu{t}M_1'$ there is
		$M_2\tu{u}M_2'$ where $l(u)=l(t)$ and $M_1'\sim_i M_2'$
and for each $M_2\tu{t}M_2'$ there is $M_1\tu{u}M_1'$ where $l(u)=l(t)$
		and $M_1'\sim_i M_2'$).
\end{itemize}
It is easy to note that $\sim_i$ are equivalences (for all $i\in\Nat$),
$\sim_0\mathop{\supseteq}\sim_1\mathop{\supseteq}\sim_2\mathop{\supseteq}\cdots$,
and
$\sim\mathop{=}\bigcap_{i\in\Nat}\sim_i$;
hence $M_1\sim M_2$ iff $M_1\sim_i M_2$ for all $i\in\Nat$
(since labeled Petri nets generate image-finite labeled transition
systems~\cite{Hennessy1985}).

While bisimilarity is undecidable for labeled Petri
nets~\cite{Jancar1995}, we note that
it is decidable for so called \emph{communication-free labeled Petri
nets}~\cite{Christensen1993}
(also know as (\textbf{1},\textbf{P})-systems~\cite{Mayr2000}, or
Basic Parallel Processes~\cite{Christensen1993Thesis}).
The communication-free Petri
nets $(P,T,W)$ satisfy $|\pre{t}| \leq 1$ for all $t\in T$, and it is
thus easy to observe that in this case bisimulation equivalence
 is a congruence w.r.t.\ the marking
addition. For our aims it is useful to look at congruences on
$\Mpp$ for finite sets $P$ in general; we do this in the next paragraph, using rather
symbols $r,s$ instead of $M$ for elements of $\Mpp$, to stress that
it is not important here whether or not $P$ is the set of places of a
Petri net.

\paragraph{Congruences on $\Mpp$.}
Given a finite set $P$ (which can be the set of places of a Petri
net), an equivalence relation $\rho$ on $\Mpp$ is a
\emph{congruence} if
\begin{center}
	$r_1\mathop{\rho} r_2$ implies $(r_1+s)\mathop{\rho}\,
	(r_2+s)$ for all $s\in\Mpp$.
\end{center}
(Hence $r_1\mathop{\rho}r_2$ and $r'_1\mathop{\rho}r'_2$ implies
$(r_1+r'_1)\mathop{\rho}\,(r_2+r'_1)$ and
$(r'_1+r_2)\mathop{\rho}\,(r'_2+r_2)$, and thus also
$(r_1+r'_1)\mathop{\rho}\,(r_2+r'_2)$.)

An important known fact is captured by the next proposition,
which says that each congruence on $\Mpp$ has a finite basis (by which
the congruence is generated); this was
an important ingredient in the decidability proof
in~\cite{Christensen1993}.
We can refer, e.g., to~\cite{Red65,Hirsh94}, but we provide
a self-contained proof for convenience.

\begin{proposition}\label{prop:finbasis}
Each congruence $\rho$ on $\Mpp$, where $P$ is a finite set, is
generated by a finite subset of $\rho$, in particular
	by the set $\rho_{\textsc{b}}$ (called a \emph{basis of}
	$\rho$) consisting of the minimal elements of the set
	$\{(r,s)\in\rho \mid r\neq s\}$ w.r.t.\ the partial
	order $\leq$, where we put $(r,s)\leq (r',s')$ if $r\leq r'$ and
	$s\leq s'$.
	(Finiteness of $\rho_{\textsc{b}}$ follows by Dickson's
	lemma.)
\end{proposition}
\begin{proof}
For the sake of contradiction, we assume that the least congruence
	$\rho'$ containing $\rho_{\textsc{b}}$ is a proper subset of~$\rho$.
	We choose a pair $(r_0,s_0)\in\rho\smallsetminus\rho'$
	such that $r_0+s_0$ is the least in the set
	$\{r+s\mid (r,s)\in\rho\smallsetminus\rho'\}$ w.r.t.\ the
	cardinality-lexicographic order $\sqsubseteq$ on $\Mpp$; we
	note that the pair $(s_0,r_0)$ also has this property.
	Since $(r_0,s_0)\in\rho\smallsetminus\rho'$, we have $r_0\neq
	s_0$ and $(r,s)\leq(r_0,s_0)$ for some
	$(r,s)\in\rho_{\textsc{b}}\subseteq\rho'$; w.l.o.g.\ we assume
	$r\sqsubseteq s$ (since otherwise we would consider the pairs
	 $(s,r)\leq(s_0,r_0)$, where $(s,r)\in\rho'$ and
	 $(s_0,r_0)\in\rho\smallsetminus\rho'$).

	We derive $(r+(s_0-s),s+(s_0-s))\in\rho'$, i.e.,
	$(r+(s_0-s),s_0)\in\rho'\subseteq\rho$; since
	$(r_0,s_0)\in\rho$, we deduce
	$(r_0,r+(s_0-s))\in\rho$.
	Since $r\sqsubseteq s$ and  $r\neq s$, we have either
	$|r|<|s|$, or $|r|=|s|$ and $r\lelex s$; this
	entails that either $|r_0+(r+(s_0-s))|<|r_0+s_0|$, or
 $|r_0+(r+(s_0-s))|=|r_0+s_0|$ and $r_0+(r+(s_0-s))\lelex r_0+s_0$.
Due to our choice of $(r_0,s_0)$ we deduce that
	$(r_0,r+(s_0-s))\in\rho'$; together with the above fact
	$(r+(s_0-s),s_0))\in\rho'$
	this entails $(r_0,s_0)\in\rho'$ -- a contradiction.
\end{proof}	

\noindent \textbf{Remark.} \\
We might also note that the above defined basis $\rho_{\textsc{b}}$ could be
made smaller by including just one of symmetric pairs $(r,s)$ and $(s,r)$.

\medskip
We note that bisimilarity, i.e.\ the equivalence $\sim$, in labeled Petri nets is not a congruence in
general (unlike the case of communication-free labeled Petri nets);
we can use the simple example in Figure~\ref{fig:rsim_sbisim}, where
 $\{X\} \sim \{Y\}$, but $(\{X\}+\{X\})\not \sim (\{Y\}+\{X\})$
 (since the firing $\{X,X\}\tu{t}\emptyset$, where $t$ is labeled with
 $b$, has no imitating firing from $\{X,Y\}$, where no $b$-labeled
 transition is enabled).

Instead of looking at subclasses where $\sim$ is a congruence, we will
look at natural equivalences refining $\sim$ that are congruences for
the whole class of labeled Petri nets. The finite-basis result
(Proposition~\ref{prop:finbasis}) might give some hope regarding the
decidability of such congruences.

\section{Resource similarity and resource bisimilarity}\label{sec:ressimbisim}

In Section~\ref{subsec:ressim} we look at the largest congruence
that refines bisimilarity (the equivalence $\sim$) in labeled Petri
nets. This relation is known as
resource similarity (denoted here by $\approx$), and it is also known
to be undecidable (which entails that its finite basis is not effectively
computable). In  Section~\ref{subsec:resbisim} we then recall another
congruence, known as resource bisimilarity (denoted here by $\simeq$).
In fact, it is the largest congruence refining $\sim$ that is a bisimulation.
In Section~\ref{sec:algor}
we show the decidability of resource bisimilarity (thus answering a question that was left open
in the literature); nevertheless, the effective computability of its finite
basis is not established by this result. (We add further comments in
Section~\ref{sec:conclusions}.)

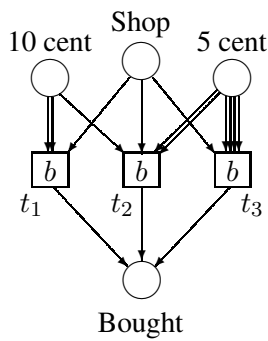
\begin{figure}[!b]
    \centering	
		\unitlength 0.9mm % = 2.561pt
\linethickness{0.4pt}
\ifx\plotpoint\undefined\newsavebox{\plotpoint}\fi % GNUPLOT compatibility
\begin{picture}(36.999,47.649)(0,0)
\put(7.324,45.126){\makebox(0,0)[cc]{10 cent}}
\put(20.673,47.649){\makebox(0,0)[cc]{Shop}}
\put(20.673,3.502){\makebox(0,0)[cc]{Bought}}
\put(34.022,45.126){\makebox(0,0)[cc]{5 cent}}
\put(7.358,39.732){\circle{5.482}}
\put(20.707,42.255){\circle{5.482}}
\put(20.812,10.091){\circle{5.482}}
\put(34.056,39.732){\circle{5.482}}
\put(4.835,23.755){\framebox(5.15,4.94)[cc]{$b$}}
\put(18.184,23.755){\framebox(5.15,4.94)[cc]{$b$}}
\put(31.534,23.755){\framebox(5.15,4.94)[cc]{$b$}}
\put(33.139,36.999){\vector(0,-1){8.094}}
\put(33.77,36.999){\vector(0,-1){8.094}}
\put(7.071,36.894){\vector(0,-1){8.094}}
\put(34.4,36.999){\vector(0,-1){8.094}}
\put(7.702,36.894){\vector(0,-1){8.094}}
\put(35.031,36.999){\vector(0,-1){8.094}}
\put(4.625,21.128){\makebox(0,0)[cc]{$t_1$}}
\put(17.974,21.128){\makebox(0,0)[cc]{$t_2$}}
\put(36.999,21.128){\makebox(0,0)[cc]{$t_3$}}
\put(20.812,39.312){\vector(0,-1){10.721}}
\put(20.812,23.65){\vector(0,-1){10.721}}
\put(9.986,28.906){\vector(-3,-4){.078}}\multiput(19.025,39.837)(-.03735387,-.045172121){242}{\line(0,-1){.045172121}}
\put(31.534,28.696){\vector(3,-4){.078}}\multiput(22.389,39.837)(.037478478,-.045663432){244}{\line(0,-1){.045663432}}
\put(18.289,28.696){\vector(1,-1){.078}}\multiput(8.829,37.315)(.041130803,-.037474732){230}{\line(1,0){.041130803}}
\put(22.599,28.801){\vector(-1,-1){.078}}\multiput(31.849,37.945)(-.037909265,-.037478478){244}{\line(-1,0){.037909265}}
\put(23.335,28.59){\vector(-1,-1){.078}}\multiput(32.269,37.63)(-.037382947,-.037822747){239}{\line(0,-1){.037822747}}
\put(18.71,12.193){\vector(1,-1){.078}}\multiput(7.883,23.755)(.0374620808,-.0400080474){289}{\line(0,-1){.0400080474}}
\put(22.704,12.088){\vector(-1,-1){.078}}\multiput(33.636,23.65)(-.0374371692,-.0395970059){292}{\line(0,-1){.0395970059}}
\end{picture}\vspace*{-3mm}
    \caption{Buying goods for 20 cents}
    \label{fig:buying}
\end{figure}

\subsection{Resource similarity}\label{subsec:ressim}

The relation of \emph{resource similarity}, introduced
in~\cite{Bashkin_FI2003} for labeled Petri nets,
is a congruent strengthening of bisimulation equivalence.
Intuitively, a resource in a Petri net can be viewed
as a part of its marking; formally it is also a multiset of places,
hence its definition does not differ from the definition of a marking.
The notions
``markings'' and ``resources'' are viewed as different
only due to their different interpretation in particular contexts. Resources
are parts of markings that may or
may not provide this or that kind of net behavior; \eg in the Petri net example
in Fig.~\ref{fig:buying} from \cite{10.1007/978-3-319-57861-3_3}, two ten-cent coins  form a resource
--- enough to buy an item of goods.

Two
resources are viewed as similar for a given labeled Petri net if replacing one of
them by another in any marking (i.e., in any context) does not change
the observable behavior of the net. In this sense, resource similarity is a
congruence that preserves visible behaviors up to bisimilarity.
Now we capture this description formally. We use symbols $r,s,w$ for
elements of $\Mpp$ (rather than $M$) to stress our ``resource
motivation'' here.

\begin{definition}
Let $N=(P,T,W,l)$ be a labeled Petri net. A \emph{resource} $r$ in $N$ is a multiset
	over the set $P$ of places; hence $r\in\Mpp$.
Two resources $r$ and $s$ in $N$ are called \emph{similar}, which is
	denoted by $r\approx s$, if for all
 $w \in\Mpp$ we have $(r+w)\sim (s+w)$, where $\sim$ is
	bisimulation equivalence.
	The relation $\approx$ is called \emph{resource similarity}.
\end{definition}

\begin{proposition}\label{prop:approxlargest}
	For any labeled Petri net $N=(P,T,W,l)$, the relation
	$\approx$ (resource similarity) is
	the largest congruence included in $\sim$ (bisimilarity).
\end{proposition}
\begin{proof}
It is straightforward to check that $\approx$ is an equivalence
	relation (reflexive, symmetric, and transitive),
since $\sim$ is an equivalence.
	(In particular, if $r_1\approx r_2$ and $r_2\approx r_3$, then
	for each $w\in\Mpp$ we have $(r_1+w)\sim (r_2+w)$ and
	$(r_2+w)\sim (r_3+w)$, and thus $(r_1+w)\sim (r_3+w)$; hence
	$r_1\approx r_3$.)
	Moreover, $\approx$
is a congruence w.r.t.\ multiset
	addition, i.e., $r\approx s$ implies $(r+w)\approx (s+w)$;
	indeed,
	for any $w'$ we have $((r+w)+w')\sim ((s+w)+w')$, since
	$r\approx s$ entails  $(r+(w+w'))\sim (s+(w+w'))$.
We trivially have
	$\approx\mathop{\subseteq}\sim$. Moreover, each congruence
	$\rho\mathop{\subseteq}\sim$ satisfies
	that $r\mathop{\rho}s$ implies
	$(r+w)\mathop{\rho}\,(s+w)$, and thus $(r+w)\sim (s+w)$, for all $w
	\in\Mpp$; this entails $\rho\mathop{\subseteq}\approx$.
\end{proof}

Since resource similarity is a congruence on $\Mpp$, it has
a~finite basis by Proposition~\ref{prop:finbasis}.
However, the place fusion  (studied in \cite{Pfister1995, Schnoebelen2000}) is a special case of resource similarity for resources of capacity one.
Place fusion has been proven to be undecidable, and this
implies undecidability of resource similarity; it also entails that
the finite basis of $\approx$ is not effectively computable.

\medskip\noindent
\textbf{Remark.}\\
The problem of deciding non-bisimilarity, i.e.\ the relation
$\not\sim$, is semidecidable
 (for labeled Petri nets); this follows from the
fact that the equivalences $\sim_i$ are decidable for all $i\in\Nat$,
and $\sim\mathop{=}\bigcap_{i\in\Nat}\sim_i$.
On the other hand, the halting problem for Minsky machines can be
reduced to deciding $\not\sim$, which entails that $\sim$ is
not semidecidable. Since $r\not\approx s$ means that $(r+w)\not\sim (s+w)$
for some $w$, we easily derive the semidecidability of
$\not\approx$. Moreover, the mentioned  reduction of the halting
problem to deciding $\not\sim$ can be adjusted so that it reduces to
deciding $\not\approx$. This entails that $\approx$ is not
semidecidable, and thus the finite basis of $\approx$ cannot be
effectively computable.

\medskip

We note that resource similarity is not a bisimulation in general;
the relation $\approx$ may not have the transfer property. This is
exemplified by the labeled Petri net in Figure~\ref{fig:rsim_rbisim}:
we shall later show that $\{X_1\}\approx\{Y_1\}$, but the firing
$\{X_1\}\tu{t_4}\{X_3\}$ has no ``imitating'' firing from $\{Y_1\}$;
indeed, both candidates $\{Y_1\}\tu{u_1}\emptyset$ and
$\{Y_1\}\tu{u_2}\{Y_2\}$ fail, since $\{X_3\}\not\approx \emptyset$
(because $(\{X_3\}+\{Z\})\not\sim_1 (\emptyset+\{Z\})$)
and  $\{X_3\}\not\approx \{Y_2\}$ (because $(\{X_3\}+\emptyset)\not\sim_1
(\{Y_2\}+\emptyset)$).

\subsection{Resource bisimilarity}\label{subsec:resbisim}

In view of the above facts, that the largest congruence $\approx$
refining $\sim$ is undecidable and is not a bisimulation, it is
natural to look at the largest congruence $\simeq$ refining $\sim$
that \emph{is} a bisimulation (whence we have
$\simeq\mathop{\subseteq}\approx\mathop{\subseteq}\sim$).
In fact, the relation $\simeq$, named \emph{resource bisimilarity},
was already defined in~\cite{Bashkin_FI2003}, though technically
slightly differently than we do this below. We also use a new term,
namely
\emph{the resource transfer property}; this property is stronger than
\emph{the transfer property} defined for relations on $\Mpp$ in
Section~\ref{sec:prelim}. An analogous notion was called \emph{the weak transfer property}
in~\cite{Bashkin_FI2003}, due to a different viewpoint (related to the
variant of the transfer property for place bisimulation defined in~\cite{Autant1992}).

\begin{definition}\label{def:resourcebis}
Let $N=(P,T,W,l)$ be a labeled Petri net.
We say that a
\emph{relation} $R\subseteq \Mpp\times\Mpp$ \emph{satisfies} the \emph{\it
resource transfer property w.r.t. a relation} $R'\subseteq \Mpp\times\Mpp$
if for each pair $(r,s)\in R$ and for each
	firing $(r+(\pre{t}-r))\tu{t}r'$ there exists an (``imitating'') firing
	$(s+(\pre{t}-r))\tu{u}s'$ such that $l(u)=l(t)$  and $(r',s')\in R'$.
We can illustrate this property by the following diagram.
$$r\qquad\quad R \qquad\quad s$$
	$$\qquad r+(\pre{t}-r)\qquad\qquad s+(\pre{t}-r)$$
$$\qquad\qquad\qquad\quad\quad\downarrow t\qquad\qquad\qquad\downarrow
	(\exists) u:\ l(u)=l(t)$$
$$r'\qquad\quad R' \qquad\quad s'$$

\medskip\noindent
(By our multiset definitions,
$r+(\pre{t}-r)=r\cup\pre{t}$.
We also note that
	$\pre{t}\leq r$ entails that $\pre{t}-r=\emptyset$, hence
	$r+(\pre{t}-r)=r$ and $s+(\pre{t}-r)=s$;
	therefore the resource transfer property
	implies the transfer property.)

\medskip
	By saying that $R$ \emph{has the resource transfer property} we mean that
$R$ satisfies the resource transfer property w.r.t. itself (hence $R'=R$ in the
diagram).

A relation $R\subseteq \Mpp\times\Mpp$ is a \emph{resource bisimulation} if
$R$ has the resource transfer property and also $R^{-1}$ has the resource transfer property.
	Resources (or markings) $r,s\in\Mpp$ are  \emph{resource
	bisimilar},
	written $r \simeq s$,
if there exists a resource bisimulation $R$ such that $(r,s)\in R$; the
relation $\simeq$, called \emph{resource bisimilarity},
 is thus the union of all resource bisimulations.
\end{definition}

\begin{example}
In the net in Figure~\ref{fig:buying} we can present
	any multiset $r$ over the set of places by
	the vector
	$(r(\texttt{10cent}),r(\texttt{Shop}),r(\texttt{5cent}),
	r(\texttt{Bought}))$.

	We can show that $r_0\simeq s_0$ where
	$r_0=(1,0,0,0)$ and $s_0=(0,0,2,0)$, by verifying that
the relation
	$\{\big((1,0,0,k),(0,0,2,k)\big)\mid k\geq 0)\}
	\cup \{\big((0,0,0,k),(0,0,0,k)\big) \mid k\geq 1\}$
	is a resource bisimulation.

	E.g.,  $\pre{t_1}-r_0=(1,1,0,0)$, and we have
	$r_0+(\pre{t_1}-r_0)=(2,1,0,0)\tu{t_1}(0,0,0,1)$.
	This firing can be imitated from $s_0+(\pre{t_1}-r_0)=(1,1,2,0)$
	by the firing $(1,1,2,0)\tu{t_2}(0,0,0,1)$ (since both $t_1$
	and $t_2$ have the same label $b$). It is a routine to finish
	this verification.
\end{example}

\begin{proposition}\label{prop:simeqlargest}
	For any labeled Petri net $N=(P,T,W,l)$, the relation $\simeq$
	(resource bisimilarity) is
	the largest congruence included in $\sim$ that is a
	bisimulation; hence we have
	$\simeq\mathop{\subseteq}\approx\mathop{\subseteq}\sim$\,.
\end{proposition}
\begin{proof}
	Since the relation $\simeq$ is the union of all resource bisimulations,
	i.e.
	\begin{center}
	$\simeq\mathop{=}\bigcup\,\{\,R\mathop{\subseteq}\Mpp\times\Mpp\mid
	R$ and $R^{-1}$ have the resource transfer property$\,\}$,
	\end{center}
	it is obvious
that also $\simeq$ and $\simeq^{-1}$ have the resource transfer
	property; hence $\simeq$
	is the largest
resource bisimulation (related to
the considered labeled Petri net $N=(P,T,W,l)$).
Since the resource transfer property is stronger
	than the transfer property (if $R$ has the resource transfer
	property, then $R$ also has the transfer property),
each resource bisimulation is a (standard) bisimulation; hence
$\simeq$ is a~bisimulation, and we have
	$\simeq\mathop{\subseteq}\sim$.

\medskip
The reflexivity and symmetry of $\simeq$ is straightforward;
the next two points show that $\simeq$ is a congruence:
	\begin{enumerate}
\itemsep=0.9pt
		\item
We prove that $r\simeq s$ implies  $(r+w)\simeq
			(s+w)$ by showing that the set
			\begin{center}
				$R=\{(r+w, s+w)\mid r,s,w\in\Mpp,
				r\simeq s\}$
			\end{center}
			is a resource	bisimulation (hence
			$R\mathop{=}\simeq$).
			To this aim we fix
some $(r+w, s+w)\in R$, where $r\simeq s$, and consider
			a firing
			$(r+w + (\pre{t}-(r+w)))\tu{t}r'$; we look
			for an imitating firing from
			$s+w +(\pre{t}-(r+w))$.
			We note that
			$w+(\pre{t}-(r+w))=(\pre{t}-r)+w'$ for a
			(maybe nonempty) multiset $w'\in\Mpp$;
			we thus have 	$(r+(\pre{t}-r)+w')\tu{t}r'$,
			and look for  an imitating firing from
			$s+(\pre{t}-r)+w'$.
			We have $(r + (\pre{t}-r))\tu{t}r''$, and
by $r\simeq s$ we deduce that there is
			an
			imitating firing $(s +
			(\pre{t}-r))\tu{t'}s''$, where
			$r''\simeq s''$.
			Since 	$(r+(\pre{t}-r)+w')\tu{t}r''+w'$ (for
			the above $r'$ we have $r'=r''+w'$), the
			firing 	$(s+(\pre{t}-r)+w')\tu{t'}s''+w'$
			satisfies our need: $(r''+w',s''+w')\in R$.
		\item
			Transitivity of $\simeq$
	(i.e., $r_1\simeq r_2$ and $r_2\simeq r_3$ entails $r_1\simeq r_3$)
is demonstrated by showing that the relation \vspace*{-2mm}
\begin{center}
$R=\{(r_1,r_3)\mid r_1\simeq r_2$ and $r_2\simeq r_3$ for some
	$r_2\in\Mpp\}$
			\end{center}
			is a resource bisimulation.
			Let us consider a pair $(r_1,r_3)\in R$, where
 $r_1\simeq r_2$ and $r_2\simeq r_3$, and
			a~firing
			$(r_1+(\pre{t_1}-r_1))\tu{t_1}r'_1$.	
			Since $r_1\simeq r_2$, there is an imitating
			firing 	$(r_2+(\pre{t_1}-r_1))\tu{t_2}r'_2$,
			where $r'_1\simeq r'_2$. By Point $1$,
			$r_2\simeq r_3$ entails
			$(r_2+(\pre{t_1}-r_1))\simeq
			(r_3+(\pre{t_1}-r_1))$, hence
			$(r_2+(\pre{t_1}-r_1))\tu{t_2}r'_2$
			(where
			$\pre{t_2}-(r_2+(\pre{t_1}-r_1))=\emptyset$)
			has an imitating firing
			$(r_3+(\pre{t_1}-r_1))\tu{t_3}r'_3$, where
			$r'_2\simeq r'_3$.
			Since $(r'_1,r'_3)\in R$, we are done.
	\end{enumerate}			
Finally we note that each congruence $\rho$ on $\Mpp$ that refines $\sim$
	(i.e., $\rho\mathop{\subseteq}\sim$)
	and is a (standard) bisimulation is also a resource
	bisimulation (since $r\mathop{\rho}s$ implies
	$(r+(\pre{t}-r))\mathop{\rho}\,(s+(\pre{t}-r))$); hence we
	have $\rho\mathop{\subseteq}\simeq$ for all such congruences.
\end{proof}

Similarly as for (standard) bisimilarity, we define the
stratified versions for resource bisimilarity; we also observe their  properties that underpin
our algorithm in Section~\ref{sec:algor}.

\begin{definition}[of equivalences $\simeq_i$ and equivalence-levels
	$\eqlev(r,s)$]
	Assuming a labeled Petri net $N=(P,T,W,l)$, we define the
	relations $\simeq_i$,  $i\in\Nat$, as follows:
\begin{itemize}
\itemsep=0.9pt
\item
	$\simeq_0\mathop{=}\Mpp\times\Mpp$;
\item
	$\simeq_{i+1}$ is the largest symmetric relation that satisfies
		the resource transfer property w.r.t. $\simeq_i$.
\end{itemize}
	For $r,s\in\Mpp$, by $\eqlev(r,s)$ we denote the largest $i$
	such that $r\simeq_i s$, stipulating  $\eqlev(r,s)=\omega$
	(where $i<\omega$ for all $i\in\Nat$)
	if $r\simeq_i s$ for all $i\in\Nat$.
\end{definition}

\begin{proposition}\label{prop:stratproperties}\hfill
	\begin{enumerate}	
\itemsep=0.9pt
		\item
For all $i\in\Nat$, $\simeq_i$ are congruences on $\Mpp$.
\item
$\simeq_0\mathop{\supseteq}\simeq_1\mathop{\supseteq}\simeq_2\mathop{\supseteq}\cdots$,
and
$\simeq\mathop{=}\bigcap_{i\in\Nat}\simeq_i$.
\item
	If $\eqlev(s,s')>\eqlev(r,s)$, then
			$\eqlev(r,s')=\eqlev(r,s)$.
	\end{enumerate}
\end{proposition}
\begin{proof}
\indent
	1) We can proceed by induction on $i$, and analogously as in
	the proof of Proposition~\ref{prop:simeqlargest}.

\smallskip
	2) The fact $\simeq_i\mathop{\supseteq}\simeq_{i+1}$ and the inclusion
	$\simeq\mathop{\subseteq}\,(\,\bigcap_{i\in\Nat}\simeq_i)$ are
	obvious. Now we observe that the relation $\bigcap_{i\in\Nat}\simeq_i$ is a resource
	bisimulation (which entails $(\bigcap_{i\in\Nat}\simeq_i)\mathop{\subseteq}\simeq$);  here we use image-finiteness, i.e., the fact
	that each firing has only finitely many candidates for imitating
	firings (with the same action-label).
	Hence if  $(r,s)\in\bigcap_{i\in\Nat}\simeq_i$, and
	$(r+(\pre{t}-r))\tu{t}r'$, then there is an imitating firing
	$(s+(\pre{t}-r))\tu{t'}s'$ (with $l(t')=l(t)$) such that
	$r'\simeq_i s'$ for infinitely many $i\in\Nat$; but this
	entails that $(r',s')\in\bigcap_{i\in\Nat}\simeq_i$.

\smallskip
	3) Let  $\eqlev(s,s')>\eqlev(r,s)=i$; hence $r\simeq_i s$,
	$r\not\simeq_{i+1} s$, and $s\simeq_{i+1}s'$ (and $s\simeq_{i}s'$).
	Since $\simeq_i$ and $\simeq_{i+1}$ are equivalences, we have
	$r\simeq_i s'$ and  $r\not\simeq_{i+1}s'$ (since
	$r\simeq_{i+1}s'$ would entail $r\simeq_{i+1}s$).
\end{proof}

The next theorem summarizes the previously derived facts on the
equivalences $\simeq$ (resource bisimilarity), $\approx$ (resource
similarity), $\sim$ (bisimilarity), and adds that these
equivalences really differ
in general.

	\begin{theorem}[Mutual relations of
		resource bisimilarity, resource similarity, and
		bisimilarity]\label{th:rsim_rbisim}
\begin{enumerate}
\itemsep=0.9pt
\item For each labeled
	Petri net we have
				$\simeq\mathop{\subseteq}\approx\mathop{\subseteq}\sim$;
				moreover, $\approx$ is the largest
				congruence included in $\sim$, and
$\simeq$ is the largest			congruence included in $\sim$
that is a bisimulation.
		\item There exists a labeled Petri net for which
			$\approx\mathop{\neq} \sim$.
		\item There exists a labeled Petri net for which
			$\simeq\mathop{\neq}\approx$.
		\item For each labeled communication-free Petri net we
			have $\simeq\mathop{=}\approx\mathop{=}\sim$.
		\end{enumerate}
\end{theorem}
	
	\begin{proof}
\indent
		Point $1$ summarizes
		Propositions~\ref{prop:approxlargest}
		and~\ref{prop:simeqlargest}.

		Point  $4$ follows by the fact that $\sim$ is a congruence for
		labeled communication-free Petri nets.

\begin{figure}[ht]
    \centering	
		\unitlength 0.9mm % = 2.561pt
\linethickness{0.4pt}
\ifx\plotpoint\undefined\newsavebox{\plotpoint}\fi % GNUPLOT compatibility
\scalebox{1.05}{\begin{picture}(57.504,8.733)(0,0)
\put(14.926,5.992){\circle{5.482}}
\put(54.763,5.992){\circle{5.482}}
\put(1.051,3.679){\framebox(5.15,4.94)[cc]{$a$}}
\put(40.888,3.679){\framebox(5.15,4.94)[cc]{$a$}}
\put(23.755,3.679){\framebox(5.15,4.94)[cc]{$b$}}
\put(11.983,5.991){\vector(-1,0){5.676}}
\put(51.82,5.991){\vector(-1,0){5.676}}
\put(17.659,6.622){\vector(1,0){5.781}}
\put(17.764,5.676){\vector(1,0){5.676}}
\put(15.031,.631){\makebox(0,0)[cc]{$X$}}
\put(54.763,.631){\makebox(0,0)[cc]{$Y$}}
\end{picture} }
    \caption{Resource similarity does not coincide with bisimilarity:
	$\{X\} \sim \{Y\}$, but $\{X\}\not \approx \{Y\}$.}
    \label{fig:rsim_sbisim}\vspace*{-2mm}
\end{figure}
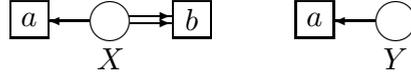

Point $2$ is shown by the example depicted in
		Fig.~\ref{fig:rsim_sbisim}.
We have 	$\{X\} \sim \{Y\}$, since the relation
		$\{ (\{X\},\{Y\}), (\emptyset,\emptyset)\}$ is a
		bisimulation; but 	$\{X\} \not\approx \{Y\}$,
		since $(\{X\}+\{X\}) \not\sim (\{Y\}+\{X\})$
		(we even have $\{X,X\} \not\sim_1 \{X,Y\}$ due to the
		action $b$ that is enabled just in one of these
		multisets).

\begin{figure}[!ht]
\vspace*{4mm}
    \centering	
\unitlength .9mm % = 2.561pt
\linethickness{0.4pt}
\ifx\plotpoint\undefined\newsavebox{\plotpoint}\fi % GNUPLOT compatibility
\scalebox{1.05}{\begin{picture}(77.783,55.507)(0,0)
\put(31.639,32.585){\circle{5.482}}
\put(17.659,32.585){\circle{5.482}}
\put(17.659,52.346){\circle{5.482}}
\put(42.15,29.011){\circle{5.482}}
\put(28.906,20.287){\framebox(5.15,4.94)[cc]{$b$}}
\put(14.926,20.287){\framebox(5.15,4.94)[cc]{$b$}}
\put(28.906,40.048){\framebox(5.15,4.94)[cc]{$a$}}
\put(14.926,40.048){\framebox(5.15,4.94)[cc]{$a$}}
\put(.946,40.048){\framebox(5.15,4.94)[cc]{$a$}}
\put(31.639,29.641){\vector(0,-1){4.31}}
\put(17.659,29.641){\vector(0,-1){4.31}}
\put(17.659,49.402){\vector(0,-1){4.31}}
\put(31.639,39.837){\vector(0,-1){4.31}}
\put(17.659,39.837){\vector(0,-1){4.31}}
\put(26.593,30.693){\makebox(0,0)[cc]{$X_3$}}
\put(12.613,30.693){\makebox(0,0)[cc]{$X_2$}}
\put(11.983,53.502){\makebox(0,0)[cc]{$X_1$}}
\put(41.94,34.582){\makebox(0,0)[cc]{$Z$}}
\put(6.096,44.988){\vector(-2,-1){.078}}\multiput(15.557,50.454)(-.06480137,-.037438356){146}{\line(-1,0){.06480137}}
\put(28.695,45.093){\vector(3,-2){.078}}\multiput(19.866,50.664)(.059255034,-.037389262){149}{\line(1,0){.059255034}}
\put(34.161,24.071){\vector(-2,-1){.078}}\multiput(39.942,27.119)(-.0705,-.03717073){82}{\line(-1,0){.0705}}
\put(3.574,37.209){\makebox(0,0)[cc]{$t_1$}}
\put(17.974,17.659){\makebox(0,0)[cc]{$t_3$}}
\put(31.954,17.554){\makebox(0,0)[cc]{$t_5$}}
\put(22.179,37.315){\makebox(0,0)[cc]{$t_2$}}
\put(36.684,37.42){\makebox(0,0)[cc]{$t_4$}}
\put(73.263,33.005){\circle{5.482}}
\put(73.263,52.766){\circle{5.482}}
\put(70.53,20.707){\framebox(5.15,4.94)[cc]{$b$}}
\put(70.53,40.468){\framebox(5.15,4.94)[cc]{$a$}}
\put(56.55,40.468){\framebox(5.15,4.94)[cc]{$a$}}
\put(73.263,30.061){\vector(0,-1){4.31}}
\put(73.263,49.822){\vector(0,-1){4.31}}
\put(73.263,40.257){\vector(0,-1){4.31}}
\put(68.217,31.113){\makebox(0,0)[cc]{$Y_2$}}
\put(67.587,53.922){\makebox(0,0)[cc]{$Y_1$}}
\put(61.7,45.408){\vector(-2,-1){.078}}\multiput(71.161,50.874)(-.06480137,-.037438356){146}{\line(-1,0){.06480137}}
\put(59.178,37.629){\makebox(0,0)[cc]{$u_1$}}
\put(73.578,18.079){\makebox(0,0)[cc]{$u_3$}}
\put(77.783,37.735){\makebox(0,0)[cc]{$u_2$}}
\end{picture} }\vspace*{-14mm}
    \caption{Resource similarity does not coincide with resource
	bisimilarity: $\{X_1\} \approx \{Y_1\}$, but $\{X_1\} \not
	\simeq \{Y_1\}$.}
    \label{fig:rsim_rbisim}
\end{figure}
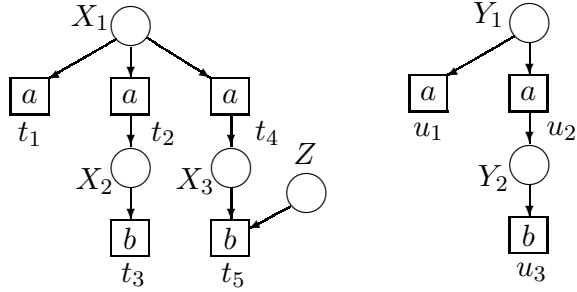

\medskip
		Point $3$ is shown by the example depicted in
		Fig.~\ref{fig:rsim_rbisim}.
At the end of Section~\ref{subsec:ressim} we have, in fact, already
		shown that  $\{X_1\}\not\simeq \{Y_1\}$: since the
		pair $(\{X_1\},\{Y_1\})$ (viewed as a relation with a
		single element)
		does not satisfy the transfer
		property w.r.t.\ $\approx$, it surely does not
		satisfy the resource transfer property w.r.t.\ $\simeq$
		(since 	$\simeq\mathop{\subseteq}\approx$, and the resource transfer property is stronger than
		the transfer property).
It thus remains to show that $\{X_1\}\approx \{Y_1\}$, i.e. that
		\begin{center}
		$(\{X_1\}+w)\sim(\{Y_1\}+w)$ for each multiset $w$
		over the set $P=\{X_1,X_2,X_3,Y_1,Y_2,Z\}$.
		\end{center}
We show this by verifying that the following relation $R$ on $\Mpp$ is
a bisimulation; $R$ contains the pairs $(r,s)$ for which one of the
following conditions is satisfied
(where the condition $1$ guarantees that $(\{X_1\}+w,\{Y_1\}+w)\in R$):
\begin{enumerate}
	\item $r(X_1)-s(X_1)=1$, $s(Y_1)-r(Y_1)=1$, and $r,s$
		coincide on the set $\{X_2,X_3,Y_2,Z\}$;
	\item $r(X_1)=s(X_1)$, $r(Y_1)=s(Y_1)$,
		and
\\
$r(X_2)+\min\{r(X_3),r(Z)\}+r(Y_2)=s(X_2)+\min\{s(X_3),s(Z)\}+s(Y_2)$.
\end{enumerate}		
It is a routine to check that both $R$ and $R^{-1}$ have the transfer
property. The only interesting cases are the firings from one side
of $(r,s)\in R$ that cannot be matched (``imitated'') by firing the
same transition from the other side. Such cases are described in what
follows:
\begin{itemize}
\itemsep=0.9pt
	\item
$(r,s)$ satisfies $2$:
\\
Here any $b$-transition (i.e., $t_3,t_5,u_3$) fired from $r$
		(or from $s$) can be matched by firing any $b$-transition
		from $s$ (or from $r$, respectively); the resulting
		pair $(r',s')$ obviously
		satisfies $2$ as well.

		For the firing $r\tu{t_4}r'$ ($t_4$ is an $a$-transition) we have
		\begin{itemize}
			\item either
		$\min\{r'(X_3),r'(Z)\}=\min\{r(X_3),r(Z)\}+1$
		(when $r(X_3)<r(Z)$),
				in which case the firing $r\tu{t_4}r'$ is matched by
		$s\tu{t_2}s'$,
	\item
		or $\min\{r'(X_3),r'(Z)\}=\min\{r(X_3),r(Z))\}$
	(when $r(X_3)\geq r(Z)$), in which case the firing
				$r\tu{t_4}r'$
				is matched by
		$s\tu{t_1}s'$.
		\end{itemize}
In both cases $(r',s')$ satisfies $2$.
		The firing $s\tu{t_4}s'$ is
		matched analogously.
		\item
			$(r,s)$ satisfies $1$ and $r(Y_1)=0$ or
			$s(X_1)=0$:		
\\
		If $r(Y_1)=0$ (hence $s(Y_1)=1$),
		then the firing $s\tu{u_1}s'$ (or
		$s\tu{u_2}s'$)
		is
		matched by  $r\tu{t_1}r'$ (or
		$r\tu{t_2}r'$, respectively); in both cases
		$(r',s')$ satisfies $2$.

If $r(X_1)=1$ (hence $s(X_1)=0$),
		then the firing $r\tu{t_1}r'$ (or
		$r\tu{t_2}r'$)
		is
		matched by  $s\tu{u_1}s'$ (or
		$s\tu{u_2}s'$, respectively).

		The firing $r\tu{t_4}r'$ is matched similarly as
		discussed above:
		\begin{itemize}
			\item
				if
		$\min\{r'(X_3),r'(Z)\}=\min\{r(X_3),r(Z)\}+1$,
				then
				 $r\tu{t_4}r'$ is matched by
		$s\tu{u_2}s'$;
	\item	if $\min\{r'(X_3),r'(Z)\}=\min\{r(X_3),r(Z)\}$,
		then
$r\tu{t_4}r'$
				is matched by
		$s\tu{u_1}s'$.
		\end{itemize}
		In both cases
		$(r',s')$ satisfies $2$.
\end{itemize}

\vspace*{-7mm}
\end{proof}

As we already mentioned, bisimilarity (the relation $\sim$)
and resource similarity (the relation $\approx$) are undecidable.
In the next section (Section~\ref{sec:algor})
we show that resource bisimilarity  (the relation $\simeq$) is decidable.
Here we still show that the decidability does not immediately follow from the fact
that $\simeq=\bigcap_{i\in\Nat}\simeq_i$
(where
$\simeq_0\mathop{\supseteq}\simeq_1\mathop{\supseteq}\simeq_2\mathop{\supseteq}
\cdots$), even though all $\simeq_i$, as well as $\simeq$, are
congruences and thus have finite bases (by
Proposition~\ref{prop:finbasis}). We could even easily derive that
finite bases for $\simeq_i$, $i\in\Nat$, are effectively computable;
nevertheless, it is not immediately clear how they ``converge''
to a~finite basis for $\sim$. The following simple example aims to
illustrate this, even for the case of a labeled communication-free Petri net
(where $\simeq$ coincides with $\sim$).

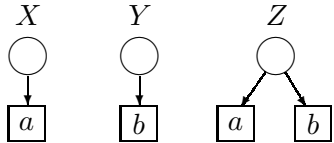
\begin{figure}[h]
\vspace*{4mm}
    \centering	
		\unitlength 0.9mm % = 2.561pt
\linethickness{0.4pt}
\ifx\plotpoint\undefined\newsavebox{\plotpoint}\fi % GNUPLOT compatibility
\begin{picture}(47.616,18.5)(0,0)
\put(3.258,12.403){\circle{5.482}}
\put(19.656,12.403){\circle{5.482}}
\put(39.627,12.403){\circle{5.482}}
\put(.526,.105){\framebox(5.15,4.94)[cc]{$a$}}
\put(16.923,.105){\framebox(5.15,4.94)[cc]{$b$}}
\put(31.113,.105){\framebox(5.15,4.94)[cc]{$a$}}
\put(42.465,.105){\framebox(5.15,4.94)[cc]{$b$}}
\put(3.258,9.46){\vector(0,-1){4.31}}
\put(19.656,9.46){\vector(0,-1){4.31}}
\put(3.364,18.5){\makebox(0,0)[cc]{$X$}}
\put(19.761,18.5){\makebox(0,0)[cc]{$Y$}}
\put(39.732,18.5){\makebox(0,0)[cc]{$Z$}}
\put(34.582,5.045){\vector(-2,-3){.078}}\multiput(38.05,10.091)(-.037297656,-.054251136){93}{\line(0,-1){.054251136}}
\put(44.042,5.045){\vector(2,-3){.078}}\multiput(40.994,10.091)(.03717361,-.06152873){82}{\line(0,-1){.06152873}}
\end{picture}
    \caption{A labeled communication-free Petri net}
    \label{fig:sequence}\vspace*{-3mm}
\end{figure}

\begin{example}
Let us consider the labeled
	Petri net in Fig.~\ref{fig:sequence}; it is a
	communication-free net, since the preset of
	each transition is a (multi)set containing a single place.
	Here bisimilarity (the relation $\sim$) is a congruence
	(thus coinciding with $\simeq$), and it is not difficult to
	derive that $\sim$ is the identity relation (we can thus take
	the empty set as its finite basis).

\medskip
	We now look at finite bases  of congruences $\simeq_i$. To
	this aim we present multisets $r$ over $\{X,Y,Z\}$ as the
	vectors $(r(X),r(Y),r(Z))$.
	For $\simeq_0$, a basis is the set
	\begin{center}
$\{\big((0,0,0),(0,0,1)\big), \big((0,0,0),(0,1,0)\big),
		\big((0,0,0),(1,0,0)\big)\}$
	\end{center}
	(where we do not include the symmetric pairs).	

\medskip
For $\simeq_1$, a basis
contains elements like
	\begin{center}
$\big((0,0,1),(0,0,2)\big), \big((0,0,1),(0,1,1)\big),
 \big((0,0,1),(1,0,1)\big), \big((0,0,1),(1,1,0)\big), \dots$
	\end{center}

For $\simeq_2$, a basis can not contain, e.g.,
	$\big((0,0,1),(0,0,2)\big)$, but
``instead''
it contains elements like $\big((0,0,2),(0,0,3)\big)$,
	$\big((1,1,1),(1,1,2)\big)$,
$\dots$.

\medskip
It seems to be a subtle problem in general, to derive a basis for
	$\simeq$ by constructing a row of bases for  $\simeq_1$,
	$\simeq_2$, $\dots$.
	Nevertheless, for labeled \emph{communication-free} Petri nets, a finite basis
	for $\sim$ (coinciding with $\simeq$) can be effectively
	computed; further remarks about this are in
	Section~\ref{sec:conclusions}.
\end{example}

\section{Algorithm for checking resource bisimilarity}\label{sec:algor}

\noindent \textbf{The problem: }\\
Given a labeled Petri net $N=(P,T,W,l)$ and two resources $r_0,s_0 \in
\Mpp$, we need to check whether they are resource bisimilar, i.e. whether $r_0 \simeq s_0$.
We assume that $P\neq\emptyset$ and $T\neq\emptyset$ (otherwise the
problem is trivial).

\medskip
The algorithm ALG we present here to solve this problem uses
the well-known tableau technique (see \eg \cite{Aceto2011}), adapted
to the resource bisimilarity relation and its resource transfer property.

\medskip

Below we give a pseudocode of ALG (with some
comments inside \algcom{}).
In its computation, ALG will be also comparing the cardinalities $|r|$, $|s|$ of
multisets $r,s\in\Mpp$, and in the case $|r|=|s|$ it will use the
\emph{lexicographic order} $r\leqlex s$.
In other words, it will use the cardinality-lexicographic order
$\sqsubseteq$ that was defined in Section~\ref{sec:prelim} for
multisets.
Besides that, ALG will also use the standard component-wise
order $r\leq s$, extended to the pairs of multisets as in
Proposition~\ref{prop:finbasis}: we put $(r,s)\leq (r',s')$
if $r\leq r'$ and $s\leq s'$.
\begin{quote}
	\textbf{Nondeterministic algorithm ALG deciding resource bisimilarity}

\emph{Input:}
A labeled Petri net $N=(P,T,W,l)$ and two resources $r_0,s_0 \in
\Mpp$.

\emph{Output:} If  $r_0 \simeq s_0$, then
at least one computation returns YES; if $r_0 \not\simeq s_0$, then
all computations return NO.

\smallskip
\textbf{Procedure:}\smallskip

\algcom{It stepwise constructs a tree (a \emph{proof tree}, also
called a \emph{tableau}), which is stored in a ``program
variable'' $\ct$ (Current Tree); its nodes are labeled with pairs
of resources. A~node of $\ct$ is called an \emph{identity-node} if its label
is of the form $(r,r)$; each such node will be a \emph{successful
leaf} of the
constructed tree. The outcome YES will be returned iff a~\emph{successful tree}, i.e.\ a tree whose all
leaves are successful, will be constructed.}
\begin{enumerate}
\itemsep=0.9pt
	\item
Create the root labeled with the input pair $(r_0,s_0)$; this node
		constitutes the initial current tree, hence the
		initial value of $\ct$.
	\item

		\textbf{while} there is a non-identity leaf
		in $\ct$ \textbf{do}

		\textbf{begin}
		
		In $\ct$ choose a leaf
		$\n$ labeled with $(r,s)$ where $r\neq s$,
		and process it as follows:
		\begin{itemize}
			\item \textbf{if} the rule REDUCE \algcom{described
				below} is applicable to $\n$
			\\
				\textbf{then}
				apply it; by doing this, $\n$ gets exactly
				one child-node $\bar{\n}$

				\algcom{
					$\bar{\n}$ becomes a new leaf in
					(the extended) $\ct$, and is labeled
				with some $(\bar{r},\bar{s})$ where
				$\bar{r}+\bar{s}$ is strictly smaller
				than $r+s$ in the
				cardinality-lexicographic order
				$\sqsubseteq$};
			\item \textbf{otherwise} \algcom{when REDUCE is not
				applicable to $\n$}, apply EXPAND to
				$\n$				
				
				\algcom{which fails if $r\not\simeq_1
				s$, in which case the
				computation returns NO, and otherwise creates
				(at most) $2\cdot |T|$ children of $\n$}.

		\end{itemize}
		\textbf{end}

		RETURN YES \algcom{here $\ct$ is successful, since all leaves
		are identity-nodes}.
\end{enumerate}

\end{quote}

To formulate the rule EXPAND, we introduce the following definitions,
referring to the underlying  labeled Petri net $N=(P,T,W,l)$.
For $t \in T$, and $r,s, r',s' \in \Mpp$, the pair $(r',s')$ is a
\emph{$t$-child of} $(r,s)$
if there is $u\in T$ such that $l(u) = l(t)$, $(r+(\pre{t}-r)) \tu{t} r'$, and
$(s+(\pre{t} -r)) \tu{u} s'$. (Recall the diagram
in Definition~\ref{def:resourcebis}.)
By $\nxt_t(r,s)$ we denote the set of all $t$-children of $(r,s)$.

\medskip
We observe a trivial fact that shows the soundness of the following
description of EXPAND.

\begin{claim}\label{cl:notsimeqone}
We have $r\not\simeq_1 s$ iff $\nxt_t(r,s)$ or  $\nxt_t(s,r)$ is
empty for some $t\in T$.
\end{claim}
\begin{quote}
	\textbf{Application of EXPAND to a node $\n$ labeled with
	$(r,s)$ (where $r\neq s$):}

	\textbf{if} $r\not\simeq_1 s$ \textbf{then} RETURN NO

\algcom{hence the algorithm ALG
invoking EXPAND returns NO in this case},

	\textbf{otherwise}  for each $t\in T$ select
	(nondeterministically)
 exactly one pair of resources from $\nxt_t(r,s)$ and exactly one pair of resources
	from $\nxt_t(s,r)$, and create (at most) $2\cdot |T|$ children
	of $\n$ whose labels
	are precisely the pairs selected for all $t\in T$.

	\algcom{We recall that $T\neq\emptyset$; hence if $r\simeq_1
	s$, then the set of children of $\n$ is nonempty.}	
\end{quote}
We note some simple facts regarding EXPAND (that are used later, in the proof of
Theorem~\ref{th:algok}):

\begin{claim}[Properties of EXPAND]\label{cl:expand}
Let $\n$ be a leaf of $\ct$, labeled with $(r,s)$ where $r\neq s$. Then we have:
	\begin{enumerate}
\itemsep=0.9pt
	\item If $r\simeq s$, then there is at least one application of EXPAND
		to $\n$ such that each arising child of $\n$ is
			labeled with an equivalent pair, i.e.\ with some $(r',s')$ where $r'\simeq
			s'$.
		\item If $\eqlev(r,s)=k\in\Nat_+$ (hence $r\simeq_k
			s$, $r\not\simeq_{k+1} s$, $k\geq 1$), then
each application of EXPAND to $\n$ gives rise to
			at least one child of $\n$ that is labeled with some
			$(r',s')$ where $\eqlev(r',s')<k$.
	\end{enumerate}
	\end{claim}
\begin{proof}
The claims easily follow from the definitions of relations $\simeq$ and
	$\simeq_i$.
\end{proof}	

Now we describe the rule REDUCE, and also note its useful properties.

\begin{quote}
	\textbf{Application of REDUCE to a node $\n$ labeled with
	$(r,s)$ (where $r\neq s$), in a given current tree $\ct$:}

If there a node $\n'\neq \n$ on the path from the root to $\n$ in
	$\ct$ that is labeled with $(r',s')$ where $(r',s')\leq
	(r,s)$, then the rule REDUCE is applicable; otherwise it is
	not applicable.	

If the rule is applicable, then $(r',s')\leq(r,s)$ related to
	one such node $\n'$ is selected, and
	$\n$ gets precisely one child,
	namely $\bar{\n}$
	labeled with $(\bar{r},\bar{s})$ where we have:
	\begin{itemize}
		\item if $|r'|<|s'|$, or  $|r'|=|s'|$ and $r'\lelex
			s'$, then $(\bar{r},\bar{s})=(r,(s-s')+r')$,
			and
	\item if $|s'|<|r'|$, or  $|r'|=|s'|$ and $s'\lelex
		r'$, then $(\bar{r},\bar{s})=((r-r')+s',s)$.
	\end{itemize}			
	\algcom{Since identity-nodes are leaves, we have $r'\neq s'$.
	But we note that the case
	$(r',s')=(r,s)$ is not excluded; in this case the child
	$\bar{\n}$
	is an identity-node, labeled with $(r,r)$ or $(s,s)$.}
\end{quote}

\begin{claim}[Properties of REDUCE]\label{cl:reduce}
Let us consider a current tree $\ct$ and an application of REDUCE to a
	node $\n$ in $\ct$, labeled with
	$(r,s)$, where the application is based on a node $\n'$ labeled with $(r',s')$
	(where $(r',s')\leq (r,s)$); let the resulting child $\bar{\n}$ of $\n$
	be labeled with $(\bar{r},\bar{s})$. We then have:
	\begin{enumerate}
\itemsep=0.9pt
		\item $|\bar{r}+\bar{s}|<|r+s|$, or
			$|\bar{r}+\bar{s}|=|r+s|$ and
			$\bar{r}+\bar{s}\lelex r+s$;
		\item
			the edges on the path from $\n'$ to $\n$ could not be
			created by applications of REDUCE only (\ie at
			least one edge has been created by applying EXPAND);
	\item
			if $r'\simeq s'$ and $r\simeq s$, then
			$\bar{r}\simeq\bar{s}$;
		\item
			if $\eqlev(r',s')>\eqlev(r,s)$, then
			$\eqlev(\bar{r},\bar{s})=\eqlev(r,s)$.
	\end{enumerate}
	\end{claim}
\begin{proof}
	1) If  $|r'|<|s'|$, then $|\bar{s}|=|(s-s')+r'|<|s|$ (since
	$s'\leq s$); hence $|\bar{r}+\bar{s}|=|r+\bar{s}|<|r+s|$.
\\
	If $r'\lelex s'$, then $(s-s')+r'\lelex s$ (due to
	the first component $i$ in which $r'$ and $s'$ differ; hence
	$(r')_i<(s')_i$, and therefore $((s-s')+r')_i<(s)_i$); this
	entails that $r+((s-s')+r')\lelex r+s$.
Hence if  $|r'|=|s'|$ and $r'\lelex s'$, then
	$\bar{r}+\bar{s}\lelex r+s$ (since $\bar{r}=r$ and
	$\bar{s}=(s-s')+r'$).
	The case $|r'|>|s'|$, or $|r'|=|s'|$ and $s'\lelex r'$ is
	analogous.

\smallskip
	2) Since $(r',s')\leq (r,s)$, we have either $|r'+s'|<|r+s|$,
	or $(r',s')=(r,s)$. By 1) it is thus obvious that the edges
	on the path from $\n'$ to $\n$ could not be all created
by applications of REDUCE.

\smallskip
	3) Since $\simeq$ is a congruence, $r'\simeq s'$ entails
	$r'+(s-s')\simeq s'+(s-s')$, hence 	$r'+(s-s')\simeq s$.
	Then  $r\simeq s$ entails $r\simeq r'+(s-s')$
	(by symmetry and transitivity of $\simeq$). The claim is thus
	clear.

\smallskip
	4) We note that $\eqlev(r',s')\leq
	\eqlev(r'+(s-s'),s'+(s-s'))$, since $\simeq_i$ are congruences
	(by Proposition~\ref{prop:stratproperties}(1)).
	Hence $\eqlev(r',s')>\eqlev(r,s)$ entails
	$\eqlev(r'+(s-s'),s)>\eqlev(r,s)$, and thus
	$\eqlev(r,r'+(s-s'))=\eqlev(r,s)$, by
	Proposition~\ref{prop:stratproperties}(3).
\end{proof}

The following theorem asserts the termination and correctness of the
algorithm ALG.

\begin{theorem}[ALG decides resource bisimilarity]\label{th:algok}
Let $N=(P,T,W,l)$ be a labeled Petri net and $r_0,s_0 \in \Mpp$ be its resources.
Then:
\begin{enumerate}
\itemsep=0.9pt
	\item
For the input $N,r_0,s_0$ there are only finitely many computations of the
above nondeterministic algorithm ALG, and each computation finishes by
		constructing a finite tree $\mathcal{T}$ whose nodes are labeled
		with pairs of resources; moreover, either one leaf of
		$\mathcal{T}$ is
		\emph{unsuccessful}, i.e.\ labeled with $(r,s)$ where
		$r\not\simeq_1 s$ (in which case the computation
		returns NO), or $\mathcal{T}$ is successful, i.e.\
		all leaves of $\mathcal{T}$ are
		identity-nodes (in which case the
		computation returns YES).
		\item
			We have	 $r_0 \simeq s_0$ if, and only if, at
			least one computation (of ALG on $N,r_0,s_0$)
constructs a~successful tree.
\end{enumerate}
\end{theorem}

\begin{proof}
1) Any constructed tree is finitely branching, since each node has at
	most $2\cdot|T|$ children. If there was an infinite
	computation, it would construct a tree with an infinite
	branch (by K\"{o}nig's lemma); let us fix such a branch
	$\textsc{b}$ for the sake of contradiction.
	Each edge on $\textsc{b}$ results by an
	application of EXPAND or REDUCE.
	Claim~\ref{cl:reduce}(1)
	entails that infinitely many edges in $\textsc{b}$ have arisen by using
	EXPAND (since we cannot have an infinite row of REDUCE
	applications).  Hence we get an infinite sequence $(r_1,s_1)$,
	$(r_2,s_2)$, $\dots$ of labels related to nodes
	$\n_1,\n_2,\dots$ in $\textsc{b}$ where EXPAND was used, and
	thus REDUCE was not
	applicable. But Dickson's lemma contradicts this, since there
	must exist some $i<j$ such that $(r_i,s_i)\leq (r_j,s_j)$
	(and thus REDUCE would be applicable to $\n_j$).

	\medskip
2) We analyze the respective two cases:
	\begin{itemize}
		\item
	Let		$r_0 \simeq s_0$.
We consider a computation which keeps the property that all nodes are
	labeled with equivalent pairs (each node is labeled with some
	$(r,s)$ where $r\simeq s$); there is such a computation by
	Claim~\ref{cl:expand}(1) and Claim~\ref{cl:reduce}(3).
			All leaves of the constructed tree are then
			successful (being identity-nodes).
\item
	Let  $r_0 \not\simeq s_0$ (hence $\eqlev(r_0,s_0)=k\in\Nat$),
			and let us consider the tree $\mathcal{T}$ constructed
			by an arbitrarily chosen computation. By recalling
			Claim~\ref{cl:expand}(2) and
			Claim~\ref{cl:reduce}(2,4) we deduce that there
			must be a branch of $\mathcal{T}$ along which
			the equivalence level never increases (it
			drops along each edge related to EXPAND, and
			remains the same along each edge related to
			REDUCE). Hence the leaf of this branch must be
			unsuccessful (labeled with some $(r,s)$ where
			$r\not\simeq_1 s$).
	\end{itemize}	

\vspace*{-6mm}
\end{proof}

\section{Conclusions and additional remarks}\label{sec:conclusions}

In this paper we have investigated the decidability issues for two
congruent restrictions of  bisimulation equivalence (denoted by $\sim$) in classical
labeled Petri nets (P/T-nets).

These congruences are resource similarity (denoted by $\approx$)
and resource bisimilarity (denoted by $\simeq$), as
defined in \cite{Bashkin_FI2003}; both try to clarify when replacing a
resource (submarking) in a Petri net marking with a similar
resource does not change the observable system behavior.

Here we have stressed that $\approx$ is the largest congruence
included in $\sim$, and that $\simeq$ is the largest congruence
included in $\sim$ that is a bisimulation; we thus also have
$\simeq\mathop{\subseteq}\approx\mathop{\subseteq}\sim$. Besides a
straightforward fact that $\approx$ is a \emph{strict} refinement of
$\sim$ in general (i.e., $\approx\mathop{\subsetneq}\sim$ for some labeled
Petri nets), we have shown here also that $\simeq$ strictly refines
$\approx$ ($\simeq\mathop{\subsetneq}\approx$); the latter fact
answered a question that was left open in~\cite{Bashkin_FI2003, 10.1007/978-3-319-57861-3_3}.
We can also notice that deciding the relations $\approx$ and $\simeq$
can be used for
deducing bisimilar states,
and thus help to increase the efficiency of verification
by reducing the state space of the relevant systems.
For another application of resource equivalences we can refer to a Petri
net reduction, in~\cite{bashkin2000reduction}.

Since resource similarity and resource bisimilarity are congruences
(w.r.t.\ addition), they are finitely-based; more specifically, they are
generated by their minimal non-identity elements.
The existence of such finite bases for $\approx$ and $\simeq$ gives
some hope at least for semidecidability of these
relations; the decidability proof in~\cite{Christensen1993} for
labeled communication-free Petri nets (under the name Basic Parallel
Processes), where we have $\simeq\mathop{=}\approx\mathop{=}\sim$, was based on such semidecidability.
Nevertheless, the previous research~\cite{Bashkin_FI2003,
10.1007/978-3-319-57861-3_3} already clarified that resource
similarity is undecidable (as well as bisimilarity),
while decidability of resource bisimilarity
remained  open.

In this paper we have shown that resource bisimilarity is decidable
for labeled Petri nets. In principle, we could proceed along the lines
of two semidecision procedures like~\cite{Christensen1993} but we have
chosen to present a tableau-based algorithm deciding the problem.
Regarding the computational complexity, we only mention the known
PSPACE-completeness of $\sim$ for Basic Parallel
Processes~\cite{Jancar22}, where $\sim\mathop{=}\simeq$.
(The PSPACE-lower bound was shown in~\cite{Srba03}.)

An interesting question that we have left open here is if the
mentioned finite basis of $\simeq$ is computable (which does not
follow from the decidability of $\simeq$). In the case of labeled
communication-free Petri nets (i.e., Basic Parallel Processes) the
answer is positive: it follows by the fact that~\cite{Jancar22} shows
that a Presburger-arithmetic description of the relation $\sim$ can be
computed in this case.
We note that for resource similarity we know that the finite basis is
not computable, since the relation $\approx$ is undecidable.

\nocite{*}
\bibliographystyle{fundam}

\end{document}